\documentclass[lettersize,journal]{IEEEtran}
\usepackage{amsmath,amsfonts}
\usepackage{algorithmic}
\usepackage{algorithm}
\usepackage{array}
\usepackage[caption=false,font=normalsize,labelfont=sf,textfont=sf]{subfig}
\usepackage{textcomp}
\usepackage{stfloats}
\usepackage{url}
\usepackage{verbatim}
\usepackage{graphicx}
\usepackage{cite}
\usepackage{amssymb}
\usepackage{cuted}

\newtheorem{definition}{\bf Definition}
\newtheorem{lemma}{\bf Lemma}
\newtheorem{proposition}{\bf Proposition}
\newtheorem{remark}{\bf Remark}
\newtheorem{proof}{Proof}

\begin{document}

\title{Age-driven Joint Sampling and Non-slot Based Scheduling for Industrial Internet of Things}

\author{Yali Cao, Yinglei Teng,~\IEEEmembership{Senior Member,~IEEE}, Mei Song, Nan Wang
\thanks{Yali Cao, Yinglei Teng,  Mei Song and Nan are with Beijing Key Laboratory of Space-ground Interconnection and Convergence, the School of Electronic Engineering, Beijing University of Posts and Telecommunications,
Beijing, 100876, China. (e-mail: lilytengtt@gmail.com)}
}

\maketitle

\begin{abstract}
Effective control of time-sensitive industrial applications depends on the real-time transmission of data from underlying sensors. Quantifying the data freshness through age of information (AoI), in this paper, we jointly design sampling and non-slot based scheduling policies to minimize the maximum time-average age of information (MAoI) among sensors with the constraints of average energy cost and finite queue stability. To overcome the intractability involving high couplings of such a complex stochastic process, we first focus on the single-sensor time-average AoI optimization problem and convert the constrained Markov decision process (CMDP) into an unconstrained Markov decision process (MDP) by the Lagrangian method. With the infinite-time average energy and AoI expression expended as the Bellman equation, the single-sensor time-average AoI optimization problem can be approached through the steady-state distribution probability. Further, we propose a low-complexity sub-optimal sampling and semi-distributed scheduling scheme for the multi-sensor scenario. The simulation results show that the proposed scheme reduces the MAoI significantly while achieving a balance between the sampling rate and service rate for multiple sensors.  
\end{abstract}

\begin{IEEEkeywords}
Industrial Internet of Things(IIoT); Time sensitive systems; Age of information(AoI); Markov decision process(MDP); URLLC
\end{IEEEkeywords}

\section{Introduction}
\IEEEPARstart{T}{he} vigorous development of 5G  enables the vision of future intelligent industrial Internet of things (IIoT) by driving operational efficiencies through trailblazing new radio (NR) Technology. Many industrial applications are part of ultra-reliable and low latency communication (URLLC) service that requires time-sensitive packet delivery, such as autonomous driving \cite{1}, intelligent manufacturing \cite{2}, and sensor networks \cite{3}, etc. These packets are periodically generated from various source nodes and transmitted to the data center through the time-varying wireless channels for information extraction and analysis to support effective industrial decision-making and control. However, restricted battery energy and communication resources may affect the acquisition of fresh information, limiting the timeliness of industrial monitoring. Therefore, how to design sampling and scheduling strategy to attain the trade-off between data freshness and resource consumption is of great significance.

Age of Information (AoI) is a metric used to quantify the freshness of data in recent years, namely the time elapsed since the generation of the last received packet at the data centre \cite{4}. The AoI includes the sampling interval and transmission delay of the packet, which represents the frequency of information being updated from the perspective of the receiver. Existing researches on sampling and scheduling problems to minimizing AoI can be classified into two categories. The first category  considers that packets are queued in the local buffer before being transmitted \cite{5,6,7,8,9,10,11}. Some studies derive the expression of Aol under different queuing rules by queuing theory. In the first-come-first-served (FCFS) system, the average AoI expression is derived in single-source \cite{5} and multi-source systems \cite{6,7,8}.
The optimal frequency of packets updates is obtained by formulating the optimal peak age of information (PAoI) problem using quasiconvex optimization \cite{7} or a derivative-free algorithm \cite{8}. In the last-come-first-served (LCFS) queuing system, the performance of average AoI is analyzed under the preemptive and non-preemptive scheme with the different distribution of service times in \cite{9}. The limitation of these studies lies in that the scheduling of sensors usually obeys the fixed probability distribution and can not realize the online adaptive transmission. Others pursue the optimal sampling and transmission strategy by establishing the relation between the AoI at the destination and the AoI at the queue. The work in \cite{10} proposes a low-complexity suboptimal policy through the linear approximation functions in relative value iteration (RVI) algorithm to overcome the curse of dimensionality caused by the increasing number of users 
with non-uniform packet sizes. In \cite{11}, the author designs a contention-based random access scheme based on the Whittle index to improve the access timeliness of sources. For the second category, there is no packet buffer, or the newly arrived packet immediately replaces the old one \cite{12,13,14,15,16,17}. Under the constraint of average power \cite{12} or energy \cite{13}, the optimal deterministic scheduling strategy is proved to have a threshold structure with respect to the AoI. The work in \cite{14} reveals that the maximum age first scheduling strategy realizes the best age performance for any given sampling rate. Then the Dynamic Programming (DP) is used to investigate the optimal sampling problem for minimizing the total average age (TaA) of sources. The author in \cite{15} expresses the multi-user scheduling  problem as a restless multi-armed bandit (RMAB) and uses the Whittle Index based approach to find the low complexity scheduling policies. Considering that the curse of dimensionality limiting the effectiveness of the solution with a growing number of sources, multi-agent deep reinforcement learning (DRL) are potential methods, such as the Deep Q Networks (DQN) \cite{18} and the Deep Determinist Policy Gradient (DDPG) \cite{19}. Nevertheless, these algorithms request BS as a centralized agent to learn an overall policy, which needs to know the status information of each source in time and brings about a heavy feedback burden with the densification of 
sources. Therefore, how to reduce communication overhead is still a tough issue.

In addition, the above researches aim at minimizing the time-average AoI. In practice, decision-making requires multiple sources to transmit packets synchronously as far as possible, e.g., the collaboration of automated vehicles, which motivates us to consider the tolerance for the worst time-average AoI in a system. Besides, few studies combine the characteristic of short packets in URLLC scenario with advanced transmission technology to further reduce the AoI, e.g., 5G New Radio (NR) \cite{20}.

According to the Third Generation Partnership Project (3GPP) R15 \cite{21}, 5G NR supports the flexible design of subframe and slot structures to shorten the duration of the transmission time interval (TTI). A slot of 14 OFDM symbols can be divided into several mini-slots with the length ranging from 1 to 13 symbols. The non-slot based scheduling scheme enables immediate transmission without waiting for the whole slot duration thus can effectively reduce the latency. Moreover, there has been some mini-slot adoption to the co-existence problem of enhanced Mobile Broadband (eMBB) and URLLC services by allocating mini-slot resources. \cite{19,22,23}. In \cite{22}, the author proposes an efficient resource allocation strategy to minimizing the data rate loss of eMBB traffic. The work in \cite{23} evaluates various scheduling schemes to optimize the throughput utility of eMBB flows while guaranteeing the ultra-low delayed demand of URLLC flows.

To quickly respond to time-sensitive services in the IIoT, we jointly design intelligent sampling and non-slot based transmission strategies to optimize the maximum time-average age of information (MAoI) among sensors. In particular, the main contributions of this paper are as follows.
\begin{itemize}
    \item An optimization framework of age-driven joint data sampling and non-slot based scheduling is proposed. To meet the ultra-low latency requirement, we adopt the non-slot based schedule pattern of 5G and design an AoI-sensitive sampling scheme. Different from previous work, we consider both energy and queue stability in the multi-sensor MAoI optimization problem to alleviate the pressure on cache, network communication, and energy consumption, which is an intractable stochastic optimization with mixed-integer programming problem (MIP).
    \item For the single-sensor case, the scheduling problem is modeled as a constrained Markov decision process (CMDP), which is transformed into an unconstrained Markov decision process (MDP) through Lagrangian relaxation. In terms of the steady-state distribution probability of the corresponding Markov chain, we derive the mathematical expressions of energy consumption and time-average AoI, which contributes to searching for the optimal sampling policy.
    \item For the multi-sensor case, we present a low-complexity sub-optimal sampling and semi-distributed scheduling scheme, which avoids the explosion of communication overhead compared with centralized scheduling in large-scale uplink networks.  
    \item The simulation results show that the packets arrival rate and scheduling delay as well as data freshness and energy consumption of each sensor are well balanced through the control of the proposed scheme. Besides, the non-slot based scheduling policy can effectively reduce MAoI and make good use of energy compared with the existing slot based scheduling policy.
\end{itemize}   
\section{SYSTEM MODEL}
Consider a real-time IIoT monitoring system composed of $\mathcal{N}= \left\{ {1,2, \cdots ,N} \right\}$ sensors and a destination, as shown in Fig.~\ref{fig1}. Each sensor samples packets periodically and sends them to the destination after temporarily storing in a local FCFS queue.
\begin{figure}[htbp]
    \centering
    \includegraphics[width=0.45\textwidth]{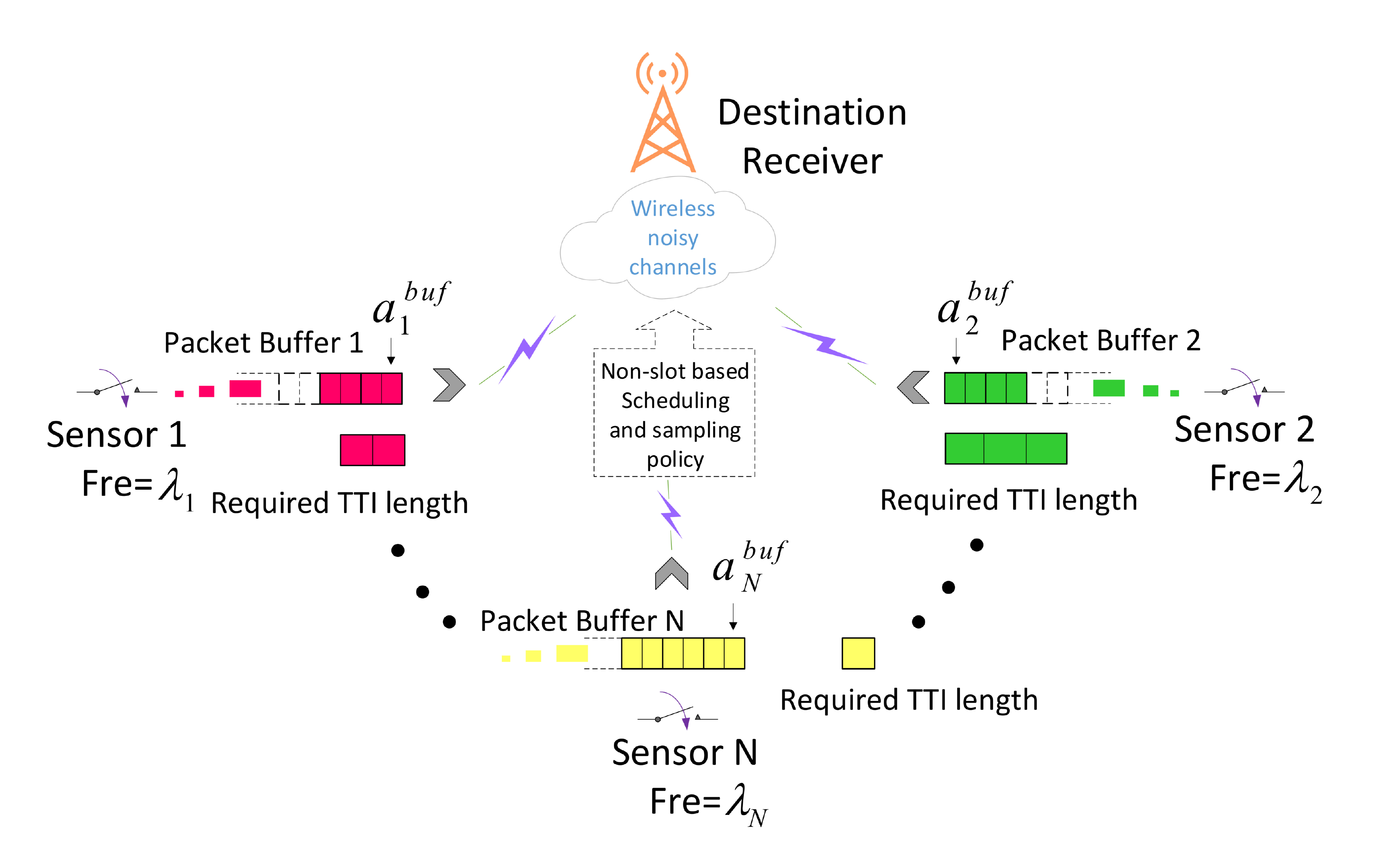}
    \caption{A real-time wireless sensor communication system.}
    \label{fig1}
\end{figure}
According to R15 in 3GPP, both flexible subframe structures and slot structures are adopted to support flexible TTI. The subcarrier spacing (SCS) is configured as 15kHz and expanded with a factor of $2^\mu$. Correspondingly, a subframe lasting 1ms is divided into $2^\mu$ slots, where $\mu\in\left\{1,2,3,4,5\right\}$. Each resource block shorten the occupancy time by expanding the frequency domain resources, which remarkably reduce the transmission delay.

Furthermore, according to 14 OFDM symbols slot arrangement \cite{21}, we divide a slot into 14 mini-slots with the length $\Delta t=1/14ms$ under the case of SCS with 15kHz \cite{21}, where the shortest length lasts for one OFDM symbol as shown in Fig.~\ref{fig2}.  Let $t=\left\{ {1,2, \cdots ,T} \right\}$ index the flexible TTI, which consists of multiple mini-slots and has a variable length. Thereby, instead of waiting for the whole duration of one slot, various mini-slots can be adaptively assigned to sensors in order to satisfy the requirements of ultra-lower transmission delay.

\begin{figure}[htbp]
    \centering
    \includegraphics[width=0.45\textwidth]{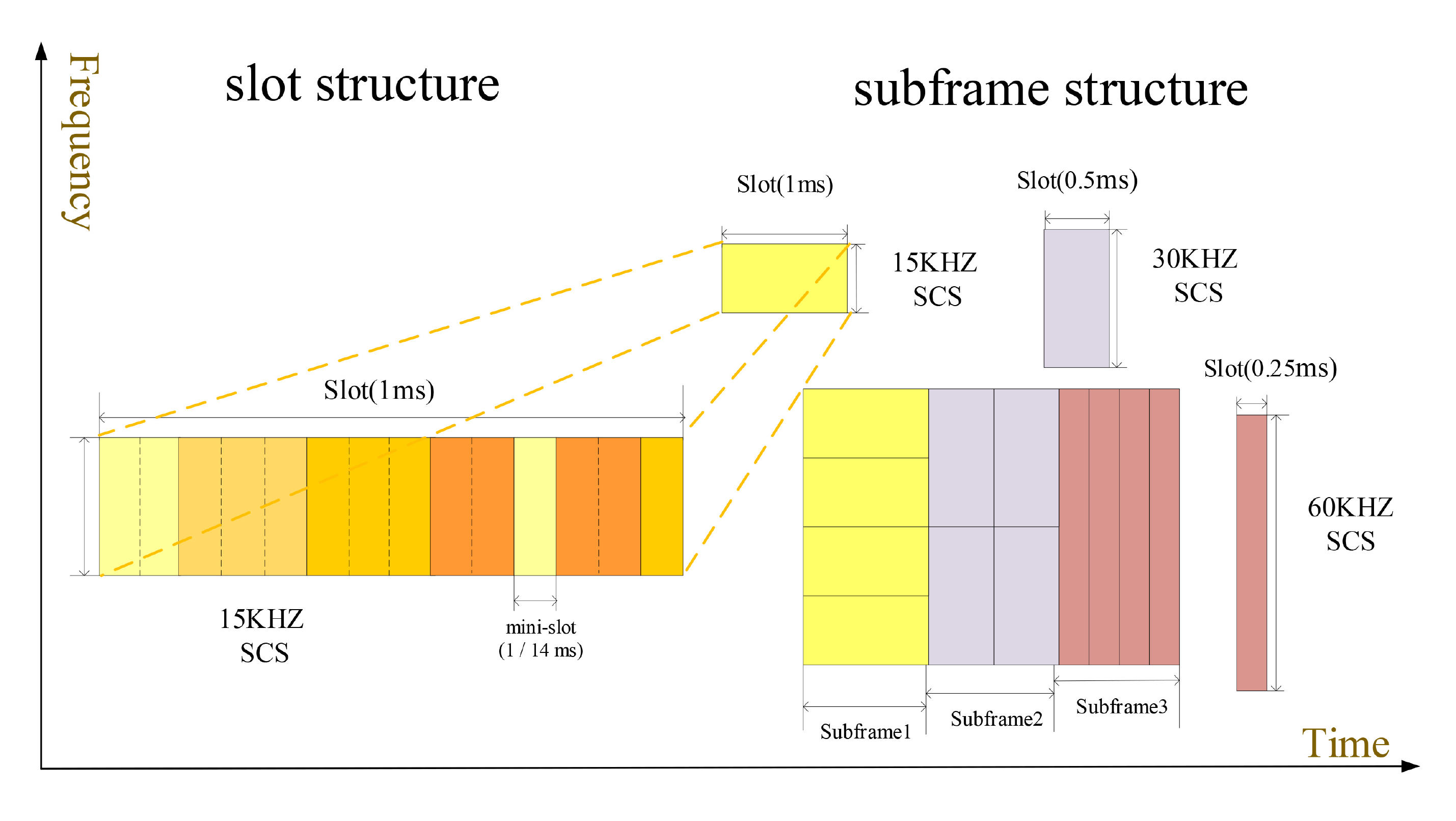}
    \caption{Division of 5G subframe structure and slot structure.}
    \label{fig2}
\end{figure}
\subsection{Transmission Model}
Define the wireless channel between the destination and each sensor with $W$-state block fading model. $W$ is a positive integer. The channel state is i.i.d across the slots and remains stable during a slot. The channel state of sensor $n$  at TTI  $t$ is denoted by $h_n(t)\in \left \{ {1,2, \cdots ,W} \right \}$. Let probability distribution of $h_n(t)$ be
\begin{equation}\label{eq1}
P_{r}\left\{h_{n}(t)=w\right\}=\alpha_{w},
\end{equation}
where $\alpha_{w}\in[0,1]$ and $\sum_{w=1}^{W} \alpha_{w}=1$. Smaller $w$ represents that channel quality is worse.

There will be conflicts if multiple sensors send packets simultaneously. Therefore, similar to \cite{11} \cite{14} \cite{16}, the central scheduler selects one sensor to transmit packets in each TTI. Let $u(t)=\{0,1 \ldots N\}$ be the schedule action of the central scheduler at the beginning of TTI $t$, $u(t)=n, n \in\{1,2 \ldots N\}$ means that the sensor $n$ is scheduled, and $u(t)=0$ represents that no sensor transmits packets. 
Consider the scenario that multiple sensors of the same type are to monitor a physical process, such as temperature and humidity monitoring. The packets sent to the destination are encapsulated in the same format. We assume that the size of each packet is $l$ bits. The length of TTI $t$ selected by the scheduled sensor is denoted by $k(t)\Delta t$, where $k(t) \in\{1,2 \ldots 14\}$ represents the amount of mini-slots consisting of the TTI. Similar to the literature \cite{10} \cite{14}, assuming that the transmission of each packet occupies a mini-slot. The number of packets that can be transmitted by the scheduled sensor $n$ is $b_{n}^{tra}(t)=k(t)$. Specially, when no sensors are scheduled, the system still needs to update the status to the next decisive moment, that is $u(t)=0, k(t)=1$ and $b_{n}^{tra}(t)=0$. Let $p_w$ denote the transmitted power to ensure that the packet is received successfully in channel state $w$, we have
\begin{equation}\label{eq2}
\frac{l}{\Delta t}=B \log \left(1+p_{w} \cdot \text{SNR}_{w}\right),
\end{equation}
 where $B$ is channel bandwidth, $\text{SNR}_{w}$ is the ratio of channel power gain to noise, e.g., $w / \delta^{2}$. Without loss of generality, $\text{SNR}_{1} < \text{SNR}_{2} \ldots \text{SNR}_{W}$, $p_{1}>p_{2}\ldots p_{W}$.
\subsection{Queue Model}
 Let $\lambda_n\in(0,1)$ indicate the sampling rate of packets of the sensor $n$, which can be periodically determined before the sensor is scheduled. The inter-arrival time between any two packets is $1 / \lambda_n$ mini-slots. Here, we consider that the sensor pre-processes the original data before caching and transmitting in some scenarios, e.g., the underlying device denoises or extracts the features of the collected image data in the camera monitoring network. There will be some associated energy cost, which limits the sampling and scheduling of the sensor. Let $c^{b}$ represent a sampling cost for generating a packet. 

With FCFS queue serving discipline, the queue model is defined as a Markov process updated as
\begin{strip}
\begin{equation}\label{eq3}
q_{n}(t+1)=\left\{\begin{array}{ll}
\min \left\{\left[q_{n}(t)+\lambda_{n} k(t)-b_{n}^{tra}\right]^{+}, q_{n}^{max}\right\}, &u(t)=n \\
\min \left\{\left[q_{n}(t)+\lambda_{n} k(t)\right]^{+}, q_{n}^{max}\right\}, & u(t) \neq n
\end{array}\right.,
\end{equation}
\end{strip}where $q_{n}(t) \in\left\{1,2 \ldots q_{n}^{max}\right\}$ represents the number of packets waiting in the queue at TTI $t$, and $q_{n}^{max}$ is queue capacity. In this paper, with jointly controlling the sample rate and transmitting action, the packet drop caused by queue overflow will not be considered.
\subsection{Age of Information Model}
We use AoI as the metric to measure the data freshness of sensor $n$ at the destination, denoted by $a_{n}^{des}(t)$, which records the time elapsed since the generation of the oldest packet of all packets that last successfully received. Considering that sensors can only transmit packets stored in the queue, the AoI at the destination can be evolved from the AoI at the queue. Let $a_{n}^{buf}(t)$ represent the AoI at the queue of the sensor $n$ at the beginning of TTI $t$, which indicates the time interval since the generation of the packet arrived earliest but not yet transmitted in the queue, namely the waiting time of the head packet. We use the discretized multiple of $\Delta t$ to evolve the AoI update process. If sensor $n$ is scheduled to transmit $b_{n}^{tra}(t)$ packets during $k(t)\Delta t$, then the AoI at the queue decrease to the time elapsed since the generation of $(b_{n}^{tra}(t)+1)th$ packet plus $k(t)$ (due to $k(t)$ mini-slots used for transmission). Otherwise, the AoI increase by $k(t)$. The dynamics of $a_{n}^{buf}(t)$ is given by
\begin{equation}\label{eq4}
a_{n}^{buf}(t+1)=\left\{\begin{array}{ll}
a_{n}^{buf}(t)+k(t)-\frac{b_{n}^{tra}(t)}{\lambda_{n}},&u(t)=n \\
a_{n}^{buf}(t)+k(t),&u(t) \neq n
\end{array}\right..
\end{equation}

If sensor $n$ sends packets, then the AoI at the destination decreases to the time elapsed since the generation of $(b_{n}^{tra}(t))th$ packet plus $k(t)$. Otherwise, the AoI increase by $k(t)$ The dynamics of $a_{n}^{des}(t)$ is given by
\begin{equation}\label{eq5}
a_{n}^{des}(t+1)=\left\{\begin{array}{ll}
a_{n}^{buf}(t)+k(t)-\frac{b_{n}^{tra}(t)-1}{\lambda_{n}},&u(t)=n \\
a_{n}^{des}(t)+k(t),&u(t) \neq n
\end{array}\right..
\end{equation}

Let $d_{n}(t) \triangleq a_{n}^{des}(t)-a_{n}^{buf}(t)$ denote the difference between the AoI at the destination and at the queue of sensor $n$, we have
\begin{equation}\label{eq6}
d_{n}(t+1)=\left\{\begin{array}{ll}
\frac{1}{\lambda_{n}}, & u(t)=n \\
d_{n}(t),  & u(t) \neq n
\end{array}\right..
\end{equation}

Obviously, $d_{n}(t)$ only have two discrete values, namely, $d_{n}(t)=\left\{0,1 / \lambda_{n}\right\}$. 
\section{PROBLEM FORMULATION AND SOLUTION}
We focus on the average age of information, energy cost and queue stability performance of the system over the infinite horizon. For sensor $n$, the time-average AoI at the destination is expressed as
\begin{equation}\label{eq7}
\overline{a_{n}^{ave\mbox{-}des}}(\mathbf{u}, \mathbf{k}, \boldsymbol{\lambda}) \triangleq \lim _{T \rightarrow \infty} \frac{1}{T} \sum_{t=1}^{T} E\left\{a_{n}^{buf}(t)+d_{n}(t)\right\}.
\end{equation}

Meanwhile, the average energy cost is defined as
\begin{equation}\label{eq8}
\overline{C_{n}^{ave}}(\mathbf{u}, \mathbf{k}, \boldsymbol{\lambda}) \triangleq \lim _{T \rightarrow \infty} \frac{1}{T} \sum_{t=1}^{T} E\left\{C_{n}(t)\right\},
\end{equation}
where $C_{n}(t) \triangleq c^{b} \lambda_{n}k(t)+p_{n}(t)k(t) \mathbb{I} \{u(t)=n\}$ is energy consumption during TTI $t$, which is composed of sampling cost and transmission cost. $\mathbb{I}(\cdot)$ is the indicator function. 

Due to the limited wireless resources, finite queues of sensors that cannot compete for the channel may become unstable. Therefore, we need to consider the queue stability region of all sensors.

Let $f_n$ denote the probability of transmission which equals the ratio of the mini-slots occupied by sensor $n$ to the total mini-slots, as follows,
\begin{equation}\label{eq9}
f_{n} \triangleq \lim _{T \rightarrow \infty} \frac{\sum_{t=1}^{T} \mathbb{I}\{u(t)=n\} \cdot k(t)}{\sum_{t=1}^{T} k(t)}.
\end{equation}

Obviously, $\sum_{n=1}^{N} f_{n} \leq 1$. Let $\tau_n$ refer to the average scheduling rate of sensor $n$ that equals the number of packets transmitted per mini-slot. Since only one packet can be transmitted in each mini slot, we have
\begin{equation}\label{eq10}
\tau_{n}=f_{n} \cdot 1.
\end{equation}

Thus, the queue is a discrete-time queue with arrive rate $\lambda_n$ and mean service rate $\tau_n$. According to Little’s law \cite{24}, the queue is stable when $\lambda_n < \tau_n$. Then we can easily deduce the following remark.
\begin{remark}\label{rem1}
A sufficient condition for the stability of all queue is $\sum_{n=1}^{N} \lambda_{n} < 1$.
\end{remark}
To improve the worst performance of AoI in the system and ensure the transmission fairness, we aim to minimize MAoI at the destination among sensors, under the constraints of average energy cost and queue stability at each sensor. The optimization problem is organized as $\mathcal{P} 0$,
\begin{strip}
\begin{equation}\label{eq11}
\begin{array}{c}
\mathcal{P}0: \min\limits_{\langle\mathbf{u}, \mathbf{k},  \boldsymbol{\lambda} \rangle}\max \left\{\overline{a_{1}^{ ave\mbox{-}des}}(\mathbf{u}, \mathbf{k}, \boldsymbol{\lambda}) \ldots \overline{a_{N}^{ave\mbox{-}des}}(\mathbf{u}, \mathbf{k}, \boldsymbol{\lambda})\right\} \\
\text { s.t. } \overline{C_{n}^{ave}}(\mathbf{u}, \mathbf{k}, \boldsymbol{\lambda}) \le C_{n}^{\max }, \forall n\in \mathcal{N} \\
\sum_{n=1}^{N} \lambda_{n}<1 \\
0<\lambda_{n}<1, \forall n\in \mathcal{N}.
\end{array}
\end{equation}
\end{strip}
Regarding the time average AoI in (7), the above problem is a complex stochastic process with both continuous and discrete variables. Moreover, the objective has no closed-form expression, which makes it impossible to solve with the traditional stochastic gradient decent optimization method. In particular, the second constraint in (11) reflects the high coupling of sampling rate among sensors and suggests that it is intractable to solve. To this end, we try to analyze the best trade-off relationship between the sampling rate and the time-average AoI through the single-sensor optimization problem, which is used to guide a uniform sampling policy that enables the problem decoupled into N single-sensor problems in section 5. To avoid ambiguity, we ignore the index of the sensor, and the singe-sensor optimization problem is $\mathcal{P} 1$,
\begin{equation}\label{eq12}
\begin{array}{ll}
\mathcal{P} 1: \min\limits_{\langle\mathbf{u}, \mathbf{k}, \lambda\rangle} \overline{a^{ave\mbox{-}des}}(\mathbf{u}, \mathbf{k}, \lambda) \\
\quad \text { s.t. } \overline{C^{a v e}}(\mathbf{u}, \mathbf{k}, \lambda) \le C^{\max } \\
\;\; \qquad 0<\lambda<1.
\end{array}
\end{equation}
However, $\mathcal{P} 1$ still involves a long-term optimization variable that keeps identical across TTIs, i.e., the sampling rate $\lambda$ as well as the shor-term strategies that need to be determined at each TTI, i.e., the scheduling decision $\mathbf{u}=\{u(1) \ldots u(T)\}$ and the length of TTI $\mathbf{k}=\{k(1) \ldots k(T)\}$, which inspires us to tackle the sampling rate and scheduling strategy separately.
\section{SINGLE-SENSOR OPTIMAL SAMPLING AND SCHEDULING}
In this section, fixing a long-term sampling rate $\lambda$, we can solve the scheduling problem by modeling it as a constrained Markov Decision Process (CMDP). Then we reformulate the sampling problem utilizing the steady-state distribution probability of the Markov chain. Finally, a bisection search method is adopted to find the optimal sampling rate.
\subsection{CMDP Formulation and Optimality Equation}
Given a typical $\lambda$, the CMDP for scheduling problems is described by a tuple $\left\langle \mathcal{S}, \mathcal{G}, P_{r}(\cdot \mid \cdot), r, c\right\rangle$.
\begin{itemize}
\item State: The state of the sensor at TTI $t$ is defined as $\mathbf{s}(t) \triangleq\left\{a^{buf}(t), d(t), q(t), h(t)\right\} \in \mathcal{S}$. Note that $a^{buf}(t)$ records waiting time of the head packet in the queue, which has an upper bound. Moreover, each of the remaining elements in the state vector takes a finite number of discrete values. Therefore, the state space is finite and countable.
\item Action: The control action of the sensor include whether to be scheduled and the length of TTI $t$, denoted by $\mathbf{g}(t) \triangleq\left\{u(t), k(t)\right\} \in \mathcal{G}$, where $u(t)=1$ represents the sensor is scheduled at TTI $t$, $u(t)=0$ otherwise.
\item Transition Probability: Let $P_{r}\{\mathbf{s}(t+1) \mid \mathbf{s}(t), \mathbf{g}(t)\}$ be the transfer probability of the system state from the current TTI $t$ to TTI $t+1$ when the action $\mathbf{g}(t)$ is adopted. Since the channel state is independent of the other elements in the state vector, the transition probability is then given by Eq.~\eqref{eq13}, where $w^{\prime}$ is the channel state at TTI $t+1$.
\item One\mbox{-}Step Reward: $\mathcal{S} \times \mathcal{G} \rightarrow R $ is the immediate reward of state-action pairs, denoted by $r(\mathbf{s},\mathbf{g})=a^{buf}+d$.
\item One\mbox{-}Step Cost: The cost function is defined as consumed energy in each state-action pair, which consists of sampling energy and transmission energy, represented as $c(\mathbf{s},\mathbf{g})=c^{b}\lambda k+pku$.
\end{itemize}
\begin{strip}
\begin{equation}\label{eq13}
\begin{array}{l}
P_{\mathrm{r}}\{\mathrm{s}(t+1) \mid \mathrm{s}(t), \mathbf{g}(t)\}=\\
P_{\mathrm{r}}\left\{a^{buf}(t+1), d(t+1), q(t+1), h(t+1) \mid \mathrm{s}(t), \mathbf{g}(t)\right\}=\\
P_{\mathrm{r}}\{h(t+1)\} P_{\mathrm{r}}\left\{a^{buf}(t+1), d(t+1), q(t+1) \mid \mathbf{s}(t), \mathbf{g}(t)\right\}=\\
\left\{\begin{array}{lr}
\alpha_{w^{\prime}},\quad \{ a^{buf}(t)+k-\frac{b^{tra}(t)}{\lambda}, \;\; \frac{1}{\lambda}, \quad \;\;q(t)+\lambda k-b^{tra}(t), \;\;w^{\prime}\}\mid \mathbf{s}(t), &\{1, k\} \\
\alpha_{w^{\prime}},\quad \{a^{buf}(t)+1, \qquad \qquad d(t), \;\;\;q(t)+\lambda,\qquad \qquad \quad w^{\prime}\} \mid \mathbf{s}(t), &\{0,1\} \\
0, &\text { otherwise }
\end{array}\right..
\end{array}
\end{equation}
\end{strip}

\begin{definition} \label{def1}
A stationary scheduling policy $\pi$ is defined as a mapping from each state $\mathbf{s}$ to the action of the sensor $\mathbf{g}$, namely $\pi : \mathcal{S} \rightarrow \mathcal{G} $.
\end{definition}

As commonly done, e.g., in \cite{12} and \cite{16}, to ensure the feasibility of the CMDP, we concentrate on stationary unichain policies. Then the average reward of any feasible scheduling policy $\pi$ over an infinite horizon is expressed as
\begin{equation}\label{eq14}
\overline{a^{ave\mbox{-}des}}(\pi \mid \lambda)=\lim_{T \rightarrow \infty} \frac{1}{T} \sum_{t=1}^{T} E_{\pi}\{r(\mathbf{s}(t), \mathbf{g}(t))\},
\end{equation}
where E(·) is the expectation taken with respect to the policy $\pi$. Meanwhile, the average cost is
\begin{equation}\label{eq15}
\overline{C^{ave}}(\pi \mid \lambda)=\lim _{T \rightarrow \infty} \frac{1}{T} \sum_{t=1}^{T} E_{\pi}\{c(\mathbf{s}(t), \mathbf{g}(t))\}.
\end{equation}

Accordingly, the goal of the infinite-horizon CMDP is to find the optimal policy $\pi$ to minimize the time-average AoI under the constraint of average energy cost as $\mathcal{P} 2$,
\begin{equation}\label{eq16}
\begin{array}{ll}
\mathcal{P} 2: \min \limits_{\pi} \overline{a^{ave\mbox{-}des}}(\pi \mid \lambda)\\
\quad \text { s.t. } \overline{C^{ave}}(\pi \mid \lambda) \le C^{max}.
\end{array}
\end{equation}

Next, to obtain the optimal policy, we transform the CMDP into an unconstrained MDP by introducing a non-negative Lagrange multiplier $y$. The average Lagrange cost is defined as
\begin{equation}\label{eq17}
\begin{array}{ll}
J_{y}(\pi \mid \lambda) \triangleq \lim \limits_{T \rightarrow \infty} \frac{1}{T} \sum_{t=1}^{T} E_{\pi}\{r(\mathbf{s}(t), \mathbf{g}(t))+\\
\qquad \qquad \qquad \qquad \qquad y c(\mathbf{s}(t), \mathbf{g}(t))\}.
\end{array}
\end{equation}
Then the optimal policy of the MDP is to minimize the Lagrange cost under the given Lagrange multiplier $y$.
\begin{equation}\label{eq18}
\pi_{y}^{*}=\arg \min _{\pi} J_{y}(\pi \mid \lambda).
\end{equation}

According to \cite[Theorem 11.7]{25} and \cite[Theorem 4.4]{26}, there exists an optimum policy for the CMDP with finite state and action space. And the solution of the CMDP has the following relationship with the MDP.
\begin{lemma} \label{lem1}
The optimal policy of the CMDP with a single constraint is a random combination of two deterministic and stationary policies, is given by
\begin{equation}\label{eq19}
\pi^{*}=\theta \pi_{y_{1}}^{*}+(1-\theta)\pi_{y_{2}}^{*},
\end{equation}
where $\pi_{y_{1}}^{*}$ and $\pi_{y_{2}}^{*}$ are the optimal policies for the unconstrained MDP with Lagrange multiplier $y_{1},y_{2}$, respectively. $\theta$ is the combination parameter, which can be calculated as follows,
\begin{equation}\label{eq20}
\theta=\frac{C^{max}-\overline{C^{ave}}\left(\pi_{y_{2}}^{*} \mid \lambda\right)}{\overline{C^{ave}}\left(\pi_{y_{1}}^{*} \mid \lambda\right)-\overline{C^{ave}}\left(\pi_{y_{2}}^{*} \mid \lambda\right)},
\end{equation}
where $\overline{C^{ave}}\left(\pi_{y_{1}}^{*} \mid \lambda\right)>C^{max}>\overline{C^{ave}}\left(\pi_{y_{2}}^{*} \mid \lambda\right)$.
\end{lemma}
Moreover, for a given $y$, the deterministic and stationary policy $\pi_{y}^{*}$ for the MDP satisfies the following Bellman equation.
\begin{equation}\label{eq21}
\begin{array}{ll}
\delta_{y}+V_{y}(\mathbf{s})=\min \limits_{\mathbf{g}}\{r(\mathbf{s}, \mathbf{g})+y c(\mathbf{s}, \mathbf{g})+\\
\qquad \qquad \qquad \qquad \sum_{\mathbf{s}^{\prime}\in \mathcal{S}} P_{r}\{\mathbf{s}^{\prime} \mid \mathbf{s},\mathbf{g}\} V_{y}(\mathbf{s}^{\prime}) \},
\end{array}
\end{equation}
where $\delta_{y}$ is the optimal average Lagrange cost, $V_{y}(\cdot)$ is the value function, and $\mathbf{s}^{\prime}$ is the next state of $\mathbf{s}$. By now, the MDP in problem (18) is derived as an equivalent optimality equation with finite transition state, which can be solved by RVI, as shown in Algorithm 1. Next, we are to tackle the steady distribution of transition probability to obtain two Lagrange multipliers.
\begin{algorithm}[htbp]
\caption{Value iteration algorithm}
\begin{algorithmic}[1]

  \STATE\textbf{Initialize}: $v^{0}(\mathbf{s})=0$ for each state $\mathbf{s}$ in $\mathcal{S}$. And set $i=0,\Delta v=0$.
  \STATE For each state $\mathbf{s}$, compute:
  \STATE $v^{i+1}=\min\limits_{\mathbf{g}}\{r(\mathbf{s},\mathbf{g})+ y c(\mathbf{s}, \mathbf{g})+\gamma\sum\limits_{{\mathbf{s}^{\prime}\in \mathcal{S}}}P_{r}\{\mathbf{s}^{\prime}\mid\mathbf{s},\mathbf{g}\}V_{y}(\mathbf{s}^{\prime})\}$;
  \STATE $\Delta v(\mathbf{s})=\max \{\Delta v,  \lvert v^{i+1}(\mathbf{s})-v^{i}(\mathbf{s}) \rvert \}$;
\IF{$\Delta v(\mathbf{s})<\zeta$ for all state $\mathbf{s}\in S$}
    \STATE go to step 10.
\ELSE 
    \STATE $i=i+1$  and return to step2.
\ENDIF 
   \STATE For each state $\mathbf{s}$, compute:
   \STATE $\pi_{y}^{*}(\mathbf{s})=\arg \min\limits_{\mathbf{g}}\{r(\mathbf{s},\mathbf{g})+ y c(\mathbf{s}, \mathbf{g})+\gamma\sum\limits_{{\mathbf{s}^{\prime}\in \mathcal{S}}}P_{r}\{\mathbf{s}^{\prime}\mid\mathbf{s},\mathbf{g}\}V_{y}(\mathbf{s}^{\prime})\}$;
\end{algorithmic}
\label{algorithm1}
\end{algorithm}
\subsection{Steady-State Distribution and Optimal Scheduling policy}
Let $M$ indicate the number of elements in the state space of the CMDP. The probability distribution of the corresponding Markov chain is denoted by a vector $\boldsymbol{\beta}=\{\beta_1,\beta_2, \cdots \beta_M\}^{T}$, where $\beta_m$ denotes the steady probability of the $m$th state. We suppose that the $m$th state is $\mathbf{s}_m=\{a^{buf},d,q,w\}$, then the probability transition matrix of states is expressed by $\mathbf{X}_{M\times M}$ as
\begin{equation}\label{eq22}
\mathbf{X}=\left[\begin{array}{lc}
\chi^{\mathbf{s}_{1} \rightarrow \mathbf{s}_{1}} & \ldots \chi^{\mathbf{s}_{M} \rightarrow \mathbf{s}_{1}} \\
\vdots & \vdots \\
\chi^{\mathbf{s}_{1} \rightarrow \mathbf{s}_{M}} & \ldots \chi^{\mathbf{s}_{M} \rightarrow \mathbf{s}_{M}}
\end{array}\right],
\end{equation}
where $\chi^{\mathbf{s}_m \rightarrow \mathbf{s}_{m^{\prime}}}$ represents the probability of transition from state $\mathbf{s}_m$ to state $\mathbf{s}_{m^{\prime}}$.Two states are defined as follows,
\begin{equation}\label{eq23}
\mathbf{s}_m^{+}=\{a^{buf}+1,d,q+\lambda,w^{\prime}\},
\end{equation}
\begin{equation}\label{eq24}
\mathbf{s}_m^{-}=\{a^{buf}+k-\frac{b^{tra}}{\lambda},\frac{1}{\lambda},q+\lambda k-b^{tra},w^{\prime}\}.
\end{equation}

For a policy $\pi_{y}^{*}$ obtained by algorithm 1, we can get the scheduling action $\mathbf{g}_{\mathbf{s}_m}=\{u_{\mathbf{s}_m},k_{\mathbf{s}_m}\}$ in any state $\mathbf{s}_m$. Then according to (13), we have
\begin{equation}\label{eq25}
\chi^{\mathbf{s}_m \rightarrow \mathbf{s}_{m^{\prime}}}=\left\{\begin{array}{lll}
(1-u_{\mathbf{s}_m})\alpha_{w^{\prime}}, &  \mathbf{s}_{m^{\prime}}= \mathbf{s}_{m}^{+}\\
u_{\mathbf{s}_m}\alpha_{w^{\prime}}, & \mathbf{s}_{m^{\prime}}= \mathbf{s}_{m}^{-}\\
0, & \text{otherwise}
\end{array}\right..
\end{equation}

There are $\mathbf{X}\boldsymbol{\beta}=\boldsymbol{\beta}$ and $\sum_{m=1}^{M}\beta_m=1$ holds from the steady-state distribution property \cite{24}. Therefore, the steady distribution $\boldsymbol{\beta}$ can be calculated with respect to the following linear equation,
\begin{equation}\label{eq26}
\left[\begin{array}{c}
\mathbf{X}-\mathbf{I} \\
\mathbf{1}^{T}
\end{array}\right] \boldsymbol{\beta}=\left[\begin{array}{l}
\mathbf{0} \\
1
\end{array}\right],
\end{equation}
where $\mathbf{I}$ is the $M$-dimensional identity matrix, and $\mathbf{1}^{T}$ is the  $M$-dimensional row vector with all the entries being 1. We rewrite Eq.(14) and Eq.(15) under the policy $\pi_{y}^{*}$ as
\begin{equation}\label{eq27}
\overline{a^{ave\mbox{-}des}}(\pi_{y}^{*} \mid \lambda)=\sum_{m=1}^{M}(a^{buf}(\mathbf{s}_m)+d(\mathbf{s}_m))\beta_m.
\end{equation}
\begin{equation}\label{eq28}
\overline{C^{ave}}(\pi_{y}^{*} \mid \lambda)=\sum_{m=1}^{M}(c^b \lambda k_{\mathbf{s}_m}+p u_{\mathbf{s}_m} k_{\mathbf{s}_m})\beta_m.
\end{equation}

Then the Lagrange multipliers $y_1$ and $y_2$ can be found by sub-gradient descent method. The overall steps are described in Algorithm 2.
\begin{algorithm}[!t]
  \caption{Optimal scheduling policy for the CMDP}
  \begin{algorithmic}[2]
   \STATE \textbf{Initialize}: $y^{(0)}=0, i=0,$ compute $\pi_{y^{(0)}}^{*} $ by Algorithm 1 and $\overline{C^{ave}}(\pi_{y^{(0)}}^{*} \mid \lambda)$ according to Eq.(28).
   \IF{$\overline{C^{ave}}(\pi_{y^{(0)}}^{*} \mid \lambda) \le C^{max}$}
    \STATE {the energy constraint is satisfied.}
   \ELSE
    	\REPEAT{
    	 \STATE $i=i+1$;
         \STATE $y^{(i)}=y^{(i-1)}+\eta (\overline{C^{ave}}(\pi_{y^{(i-1)}}^{*}\mid \lambda)-C^{max})$;
         \STATE compute $\pi_{y^{(i)}}^{*}$ and $\overline{C^{ave}}(\pi_{y^{(i)}}^{*} \mid \lambda)$;
    	  }
    	 \UNTIL{$\overline{C^{ave}}(\pi_{y^{(i)}}^{*} \mid \lambda) < C^{max}$ and $\lvert y^{(i)}-y^{(i-1)}\rvert <\varepsilon$ or $i>i_{stop}$.}
	\ENDIF 
    \IF{$i<i_{stop}$}
    \STATE $\theta=\frac{C^{max}-\overline{C^{ave}}\left(\pi_{y^{(i)}}^{*} \mid \lambda\right)}{ \overline{C^{ave}}\left(\pi_{y^{(i-1)}}^{*} \mid \lambda\right)-\overline{C^{ave}}\left(\pi_{y^{(i)}}^{*} \mid \lambda\right)}$;
    \STATE $\pi^{*}=\theta\pi_{y^{(i-1)}}^{*}+(1-\theta)\pi_{y^{(i)}}^{*}$;
    \STATE $\overline{a^{ave-des}}(\pi^{*}\mid\lambda)=\theta \overline{a^{ave-des}}(\pi_{y^{(i-1)}}^{*}\mid\lambda)+(1-\theta)\overline{a^{ave-des}}(\pi_{y^{(i)}}^{*}\mid\lambda)$;
      
      \ELSE
         \STATE there is no the optimal policy.
      \ENDIF 
\end{algorithmic}
\label{algorithm2}
\end{algorithm}
\subsection{Optimal Sampling Rate}
Based on the steady distribution probability of the Markov chain, the problem of sampling rate can be expressed as $\mathcal{P} 3$,
\begin{equation}\label{eq29}
\begin{array}{ll}
\mathcal{P} 3: \min\limits_{\lambda}\sum_{m=1}^{M}(a^{buf}(\mathbf{s}_m)+d(\mathbf{s}_m))\beta_m\\
\quad \text { s.t. } \sum_{m=1}^{M}(c^b \lambda k_{\mathbf{s}_m}+p u_{\mathbf{s}_m} k_{\mathbf{s}_m})\beta_m \le C^{max}\\
\qquad \qquad0<\lambda<1.
\end{array}
\end{equation}
\begin{proposition}\label{pro1}
For $\mathcal{P} 3$, there exists the unique $\lambda^{*}$, when $\lambda<\lambda^{*}$, the time-average AoI decreases with the increase of the sampling rate, when $\lambda>\lambda^{*}$, time-average AoI increases sharply.
\end{proposition}
\begin{proof}
 When $\lambda$ is small, the energy and cache capacity of the sensor for transmission are sufficient, the main factor affecting the AoI is the sampling rate of packets. When $\lambda$ continues to increase, the average energy cost is close to the upper bound. The energy constraint will limit the transmission of packets, and the increased sampling rate no longer reduces the AoI. In particular, there may be no policy to satisfy the energy constraint at a very high sampling rate, e.g., the iterations exceed the maximal number $i_{stop}$ in algorithm 2. This indicates that the transmission status and cache capacity of the sensor are not enough to meet such a high sampling rate. In this case, the queue will be unstable, and the AoI at the destination will accumulate over time. The simulation result confirms our discussion in Fig~\ref{fig4}.
 \end{proof}
 
Then we can find the optimal $\lambda$ with the minimum time-average AoI by the bisection search method. The infinite state space of the CMDP caused by the infinite sampling rate makes the problem intractable. Thus, we discretize $\lambda$ to $\{0.1, 0.2, \cdots, 1\}$ in Algorithm 3.
\begin{algorithm}[htbp]
  \caption{Optimal sampling rate}
  \begin{algorithmic}[1]
   \STATE \textbf{Initialize}: $\lambda_{l}=0.1,\lambda_{m}=1$, compute $\pi^{*}$ and $\overline{a^{ave\mbox{-}des}}(\pi^{*}\mid\lambda_{l})$ by Algorithm 2. 
    \REPEAT{
      \STATE $\lambda=(\lambda_{l}+\lambda_{m})/2$ and accurate to 0.1;
      \STATE compute $\pi^{*}$ and $\overline{a^{ave\mbox{-}des}}(\pi^{*}\mid\lambda)$;
       \IF{$\pi^{*}$ exists and $\overline{a^{ave\mbox{-}des}}(\pi^{*}\mid\lambda)<\overline{a^{ave\mbox{-}des}}(\pi^{*}\mid\lambda_{l})$}
         \STATE  $\lambda_{l}=\lambda$;
        \ELSE 
	     \STATE  $\lambda_{m}=\lambda$;
       \ENDIF 
       }
     \UNTIL{$\rvert\lambda-\lambda_{l}\rvert<0.1$ or $\rvert\lambda_{m}-\lambda\rvert<0.1$}
   \STATE $\lambda^{*}=\lambda_l$;
   \end{algorithmic}
\label{algorithm3}
\end{algorithm}
\section{MULTI-SENSOR SUB-OPTIMAL SAMPLING AND SEMI-DISTRIBUTED SCHEDULING SCHEME}
In the multi-sensor case, we are to minimize the MAoI among sensors in the system. That is equivalent to keeping the AoI of each sensor as small as possible.

To solve the optimal sampling rate $\boldsymbol{\lambda^{*}}=\{\lambda_1^{*}, \lambda_2^{*}, \cdots, \lambda_N^{*}\}$, we can iteratively list different sampling rate combinations $\boldsymbol{\lambda}$ and compare corresponding MAoI to find the optimal value, which has the curse of dimensionality limiting the effectiveness of the solution with a growing number of sensors. To address this issue, we note that algorithm 3 ensures that the designed sampling rate will not exceed the turning point of rapid deterioration of AoI in Figure~\ref{fig4}. Therefore, the sampling rate is insensitive to the upper bound value in the feasibility domain, and a uniform sampling policy is a sub-optimal solution realizing the consistent data freshness of each sensor as far as possible.

\textbf{Uniform sampling bound}: The sampling rate of each sensor has the same upper bound in a system where only one of $N$ sensors can transmit packets at each TTI, e.g., $1 / N$ .

Afterward, we converted $\mathcal{P} 0$ as follows,
\begin{equation}\label{eq30}
\begin{array}{c}
\mathcal{P} 4: \min\limits_{\langle\mathbf{u}, \mathbf{k}, \lambda \rangle}\max \left\{\overline{a_{1}^{ ave\mbox{-}des}}(\mathbf{u}, \mathbf{k}, \boldsymbol{\lambda}) \ldots \overline{a_{N}^{ave\mbox{-}des}}(\mathbf{u}, \mathbf{k}, \boldsymbol{\lambda})\right\} \\
\text { s.t. } \overline{C_{n}^{ave}}(\mathbf{u}, \mathbf{k}, \boldsymbol{\lambda})<C_{n}^{\max }, \forall n\in \mathcal{N} \\
0<\lambda_{n}< \frac{1}N, \forall n\in \mathcal{N}.
\end{array}
\end{equation}

The above problem can be decomposed into $N$ single-sensor problems that are solved distributed according to the Algorithm 1-3. To meet the constraint that only one sensor can be scheduled in each TTI, we provide a semi-distributed scheduling policy.

\begin{algorithm}[!b]
  \caption{Semi-distributed Scheduling}
  \begin{algorithmic}[1]
   \STATE \textbf{Initialize}: $t=0$. The initial position of the mini-slot in a slot $i=1$. The AoI of all sensors have been synchronized initially, namely, $a_n^{buf}(0)=0$,$a_n^{des}(0)=0, \forall n$. Solving $\mathcal{P} 1$ to obtain an off-line sampling and scheduling policy for each sensor. At the beginning of TTI $t$:
   \STATE \textbf{Updating policy of all sensors}:All sensors update scheduling actions locally, and some report their state and actions $\{a_n^{des}(t), u_n(t), k_n(t) \mid u_n(t)=1$ and $k_n(t)+i<14\}$ to the central scheduler.
   \STATE \textbf{Mini-slot allocation at the central scheduler}:If the set $\Omega(t)$ is not empty, the central scheduler will select the sensor $n^{\prime}$ with the largest AoI and broadcast the decision $u(t)=n^{\prime}, k(t)=k_{n^{\prime}}(t)$ to all sensors.
   \STATE \textbf{Updating state of system}:
   {
       \IF{$i+k(t)= 14$}
        \STATE{$i=1$;}
       \ELSE
         \STATE { $i=i+k(t)$;}
       \ENDIF
       \STATE \textbf{} Update states of all sensors as (3)-(6) and return to step 2.
   }
\end{algorithmic}
\label{algorithm4}
\end{algorithm}
\textbf{Semi-distributed scheduling scheme}: At the beginning of $t$, sensors locally determine scheduling actions based on their state according to the solution of $\mathcal{P} 1$. Let $\Omega$ be the set of sensors that need to send packets and their channel is stable within the required TTI length. Only these sensors will send their state and scheduling actions to the central scheduler. Given the above sampling strategy, the worst-case AoI should be scheduled to avoid excessive AoI accumulation. Then the sensor with the largest AoI in $\Omega(t)$ is selected and assigned with the corresponding number of mini-slots. This is consistent with the greed strategy, which is proved to be the optimal policy to minimize AoI performance in \cite{27}.

In this way, the central scheduler does not need to know all states of all sensors, which avoids huge communication overhead and is friendly to large-scale IoT networks. More details are given in algorithm 4. 
\section{Simulation}
This section provides numerical results to illustrate the performance of the proposed scheme. The related parameters are summarized in Table 1.

To verify the performance of our proposed scheme, we compare it with the following baselines.

1) \emph{Proposed scheme without sampling control (Proposed scheme without SC)}: To meet the constraints of energy consumption and queue stability, each sensor generates packets at the minimum sampling rate, e.g., 0.1.

2) \emph{Slot based scheduling without sampling control (Slot based scheduling without SC)}: This scheme is the same as the greedy strategy in [27], that is, each sensor generates one packet per slot, and the sensor with the largest AoI that has transmission requirements is scheduled to send a packet at the slot boundary rather than the mini-slot. To be fair, we set the maximum number of packets being transmitted during a slot as 14 while keeping the fixed sampling rate consistent with baseline 1.

3) \emph{Slot based scheduling with sampling control (Slot based scheduling):} Based on slot based scheduling scheme, the control of sampling rate as the proposed scheme is adopted.
\begin{table}[htbp]
\centering
\caption{simulation parameters}
\scriptsize
\begin{tabular}{cccc}\hline
System Parameters & values\\\hline
The length of a slot & $1ms$\\
The length of a mini-slot & $\Delta t=1/14ms$\\
Channel Bandwidth & $B=180k Hz$\\
Channel States & $W=5$\\
Channel Probability Distribution & $\alpha_{w}=1/W,\forall w$\\
Signal to Noise Ratio & $\{-20,-10,0,10,20\}dB$\\
Packet Size & $l=8bits $\\
Sampling Cost & $c^b=1J $\\
Discount Factor in Algorithm 1& $\gamma=0.95$ \\
Convergence Parameters in Algorithm 1 & $\zeta=0.01$\\
Convergence Parameters in Algorithm 2 & $\varepsilon=0.01$\\\hline
    \end{tabular}
\label{tab1}
    \end{table}
\subsection{The Single-sensor Case}
\begin{figure}[!t]
    \centering
    \includegraphics[width=0.45\textwidth]{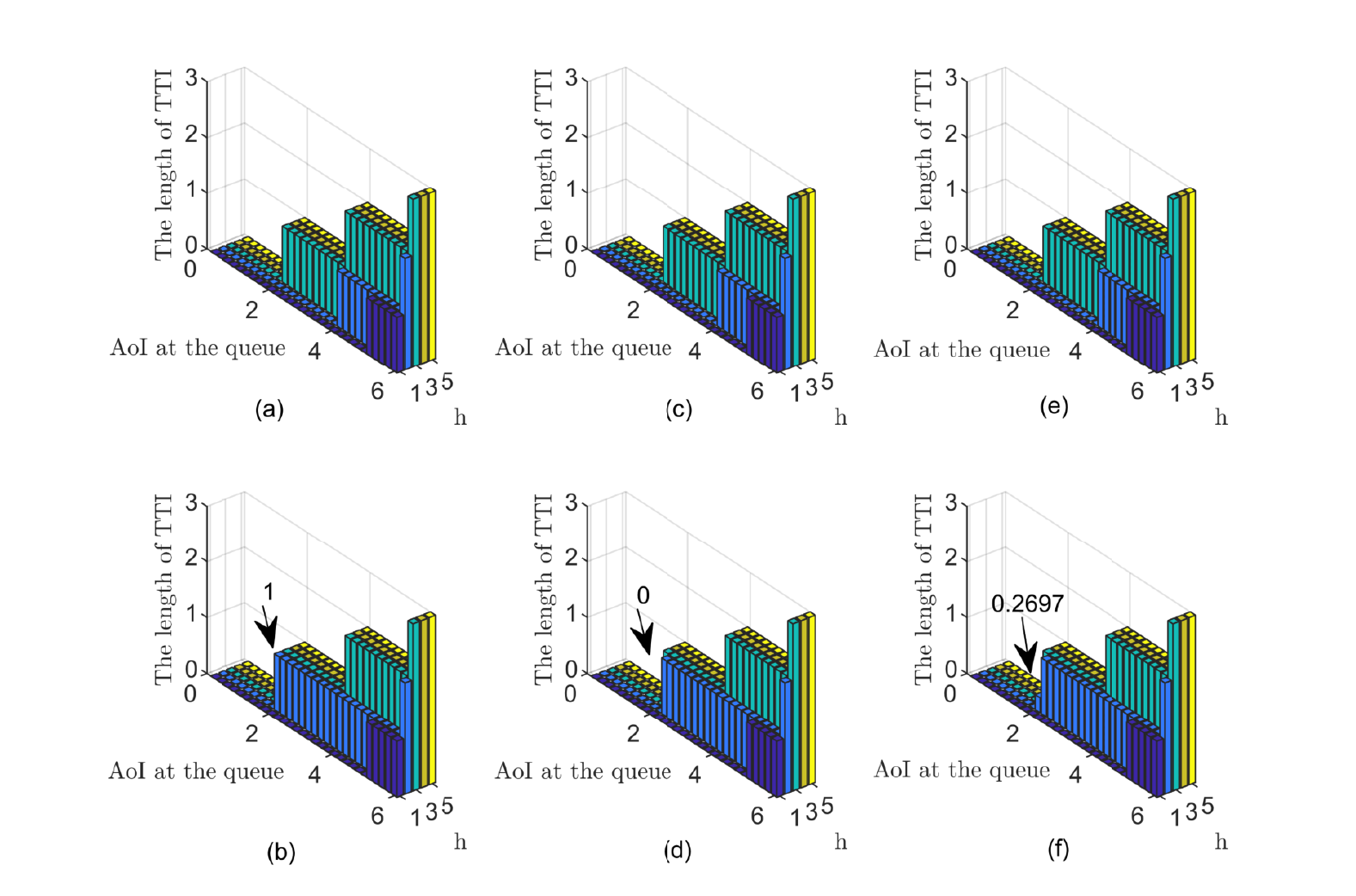}
    \caption{Two deterministic policies and the optimal policy.
$\textbf{(a)} \pi(y_1^*), d=0$ $\textbf{(b)} \pi(y_1^*), d=1 / \lambda$\\
$\textbf{(c)} \pi(y_2^*), d=0$ $\textbf{(d)} \pi(y_2^*), d=1 / \lambda$ \\
$\textbf{(e)} \pi(y^*), d=0$ $\textbf{(f)} \pi(y^*), d=1 / \lambda$}
\label{fig3}
\end{figure}
\begin{figure}
\centering
\includegraphics[width=0.45\textwidth]{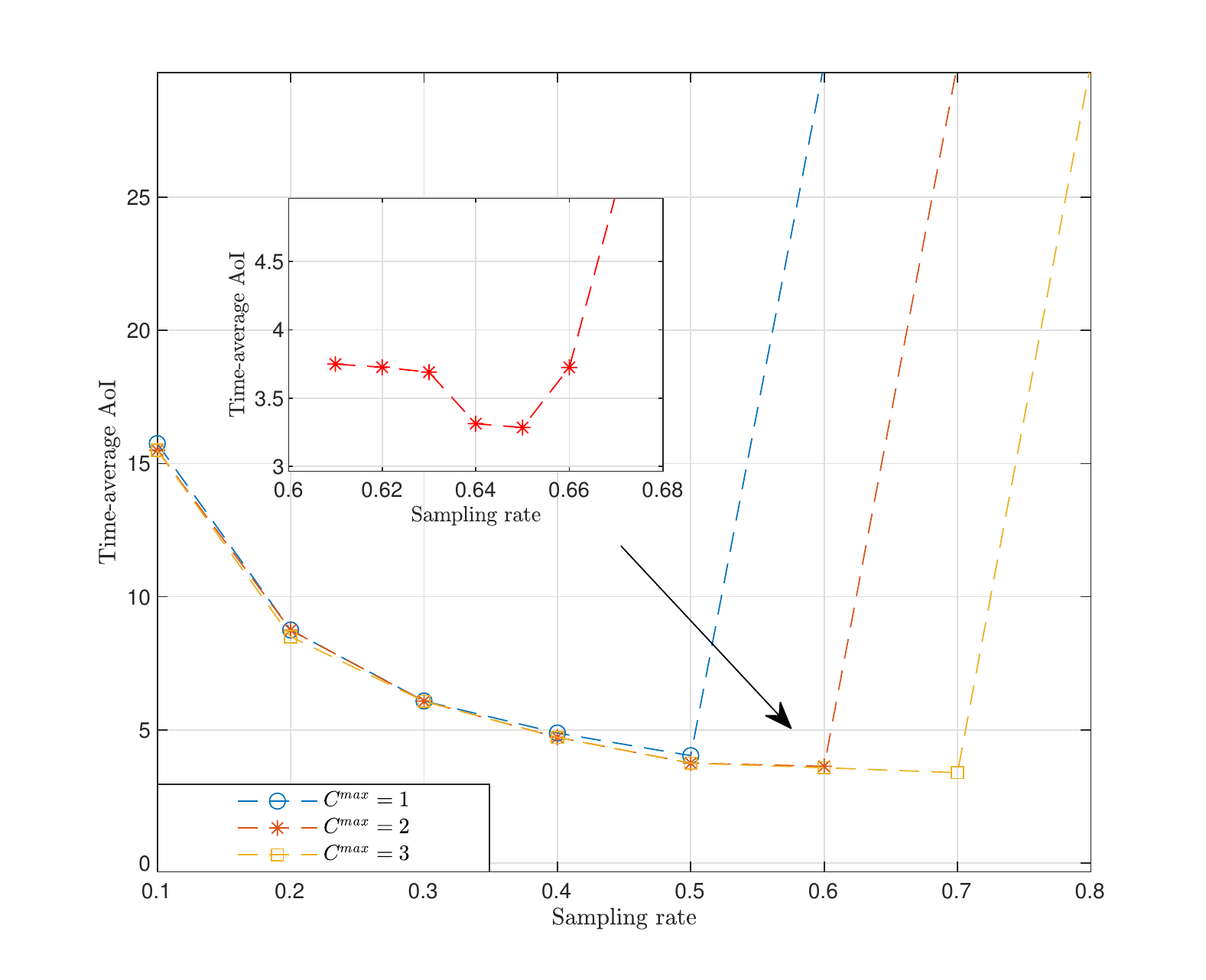}
\caption{Time-average AoI performance as sampling rate $\lambda$.}
\label{fig4}
\end{figure}
We assume that the length of the queue is 3 packets, the sampling rate is 0.5 packets per mini-slot, and the average energy constraint $C^{max}=1J$. Fig.~\ref{fig3} visualizes two deterministic and stationary policies $\pi_{y_1}^{*}$, $\pi_{y_2}^{*}$, as well as the optimal policy $\pi_{y}^{*}$. The z-axis is the product of scheduled actions $(u_{\mathbf{s}} \cdot k_{\mathbf{s}})$ at state $\mathbf{s}$, where 0 denotes the sensor does not send packets, the non-zero number represents the length of TTI when the sensor sends packets. It shows that two deterministic and stationary policies are different in one state, that is, the (2, 2) indicated by the black arrow. Meanwhile, the TTI length of the optimal scheduling policy in this state is 0.2397, which is the probability of the sensor selecting a mini-slot for transmitting. Moreover, for fixed $d$ and $h$, we note that the sensor chooses longer TTI when AoI at the queue is large. And the worse channel condition is, the sensor starts transmitting when $a^{buf}$ is greater. This indicates that in order to avoid consuming a large of transmission energy, the sensor waits until staleness cannot be tolerated anymore. When the channel quality is good, the relatively sufficient energy enables packets to be sent timely to maintain the freshness of the data.
\begin{figure}[!t]
\centering
\includegraphics[width=0.5\textwidth]{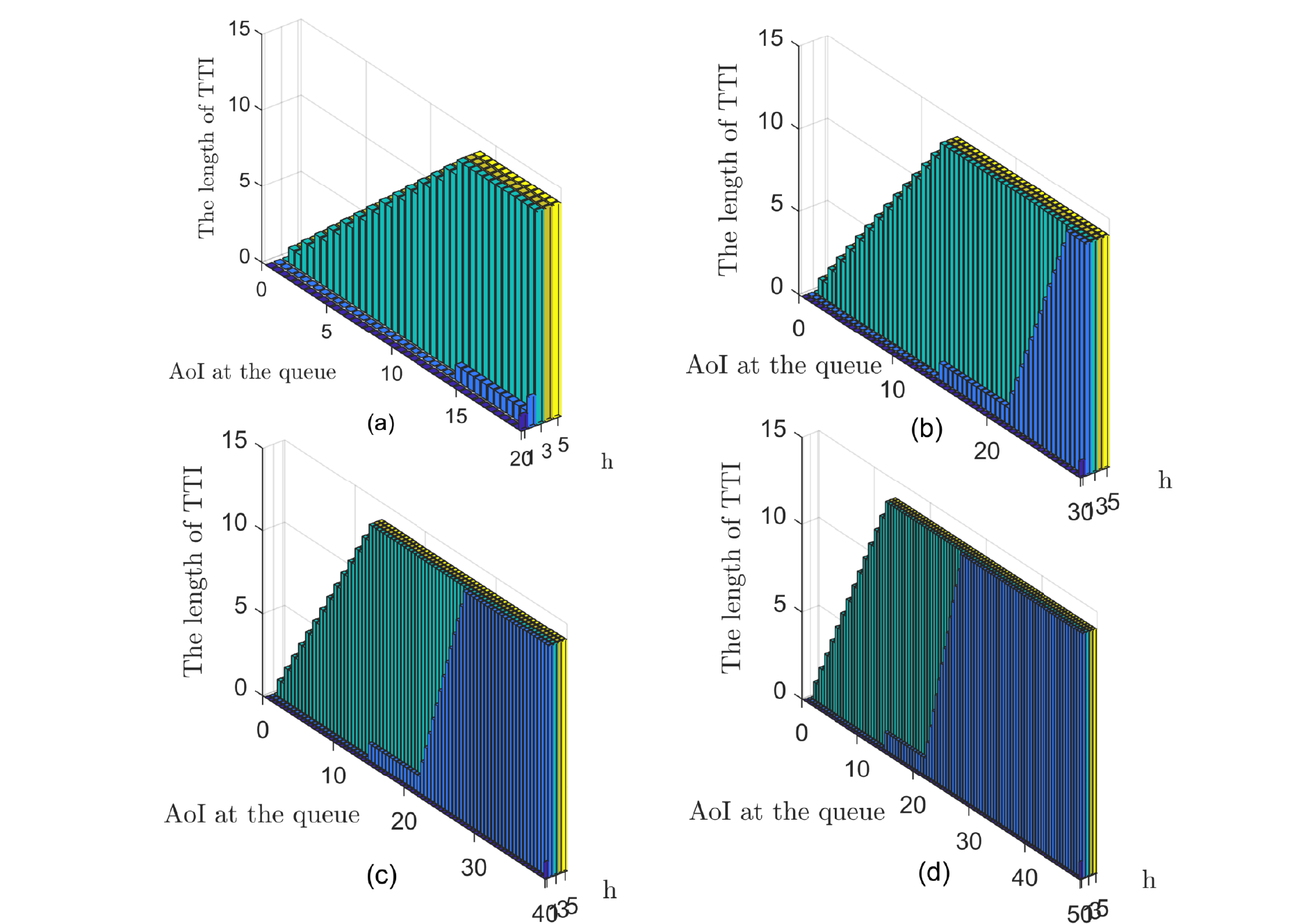}
\includegraphics[width=0.5\textwidth]{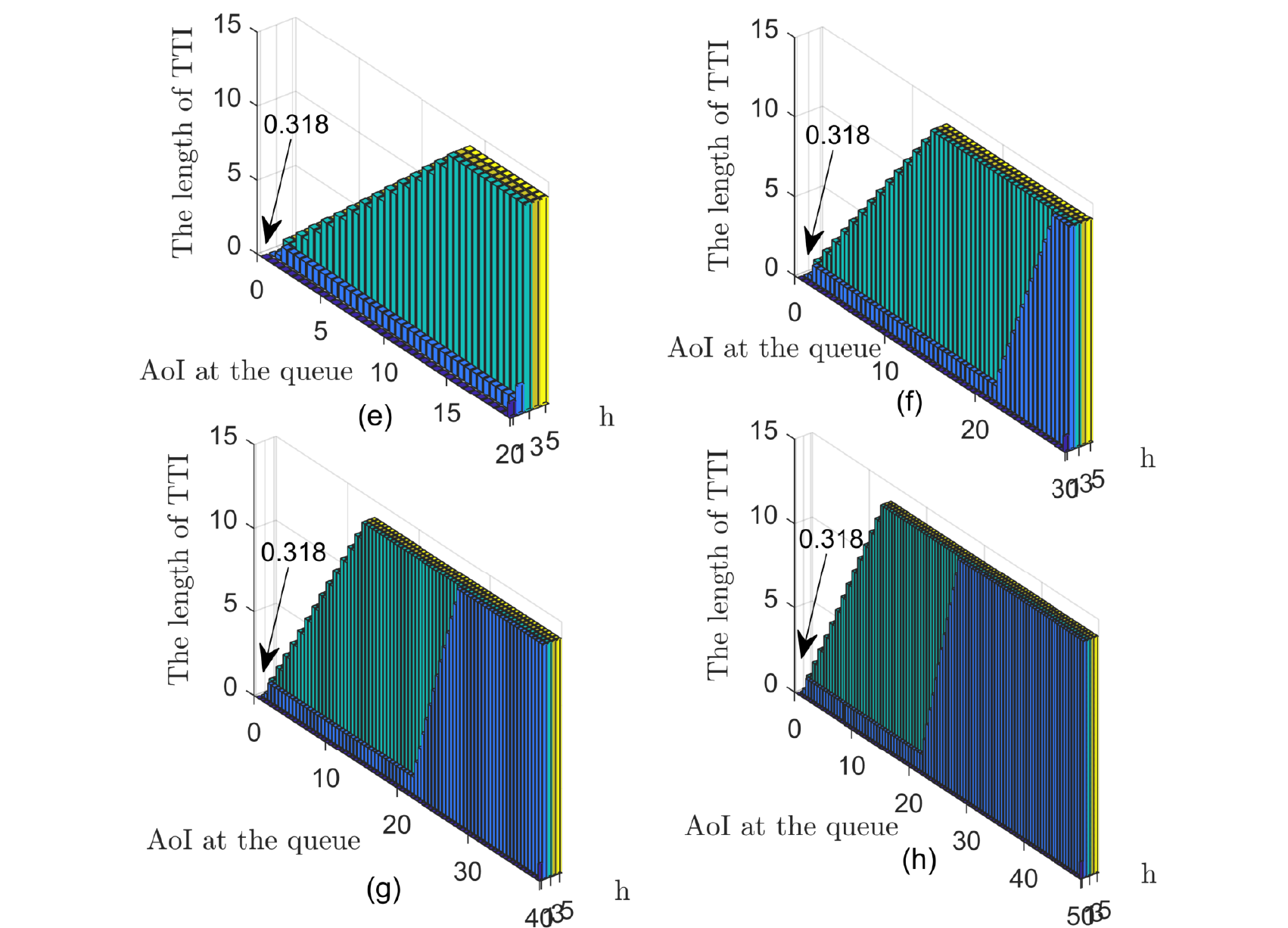}
\caption{The optimal policy as queue capacity $q^{max}$.
$\textbf{(a)} q^{max}=20, d=0$ \qquad $\textbf{(b)} q^{max}=30, d=0$ \\
$\textbf{(c)} q^{max}=40, d=0$ \qquad $\textbf{(d)}q^{max}=50, d=0$ \\
$\textbf{(e)} q^{max}=20, d=1 / \lambda$ \quad $\textbf{(f)} q^{max}=30, d=1 / \lambda$ \\
$\textbf{(g)} q^{max}=40, d=1 / \lambda$ \quad $\textbf{(h)} q^{max}=50, d=1 / \lambda$}
\label{fig5}
\end{figure}

Fig.~\ref{fig4} describes the relationship between the time-average AoI at the destination and the sampling rate under different energy constraints. Intuitively, as $\lambda$ increases, the time-average AoI first decreases. When sampling rate exceeds a certain threshold, the time-average AoI approaches infinity. This is because limited energy and cache capacity cannot provide high-rate sampling and transmission tasks. Then there is no feasible scheduling policy, which leads to AoI accumulates over time. To observe the trend of AoI in more detail, we further improve the accuracy of the $\lambda$ to 0.01 at the turning point of the threshold.

In Fig.~\ref{fig5}(a) - Fig.~\ref{fig5}(h), we depict the effect of queue capacity on the scheduling policy. When $q^{max}$ is larger, that is, the AoI at the queue $a^{buf}$ is bigger. We see that the sensor always keeps waiting under the worst channel state to avoid consuming a lot of transmission energy before the queue is almost full, then it has to choose the shortest TTI (a mini-slot) to transmit a packet. In other channel states, the length of TTI rises gradually with the increase of the AoI at the queue. When $a^{buf}$ reaches a threshold, the whole slot length (14 mini-slots) is always selected. This is due to the fact that packets have been waiting in the queue for too long, the increasing age of information makes the sensor choose the longest TTI to transmit as many as possible packets to improve data freshness.
\begin{figure}
	\centering
	\includegraphics[width=0.45\textwidth]{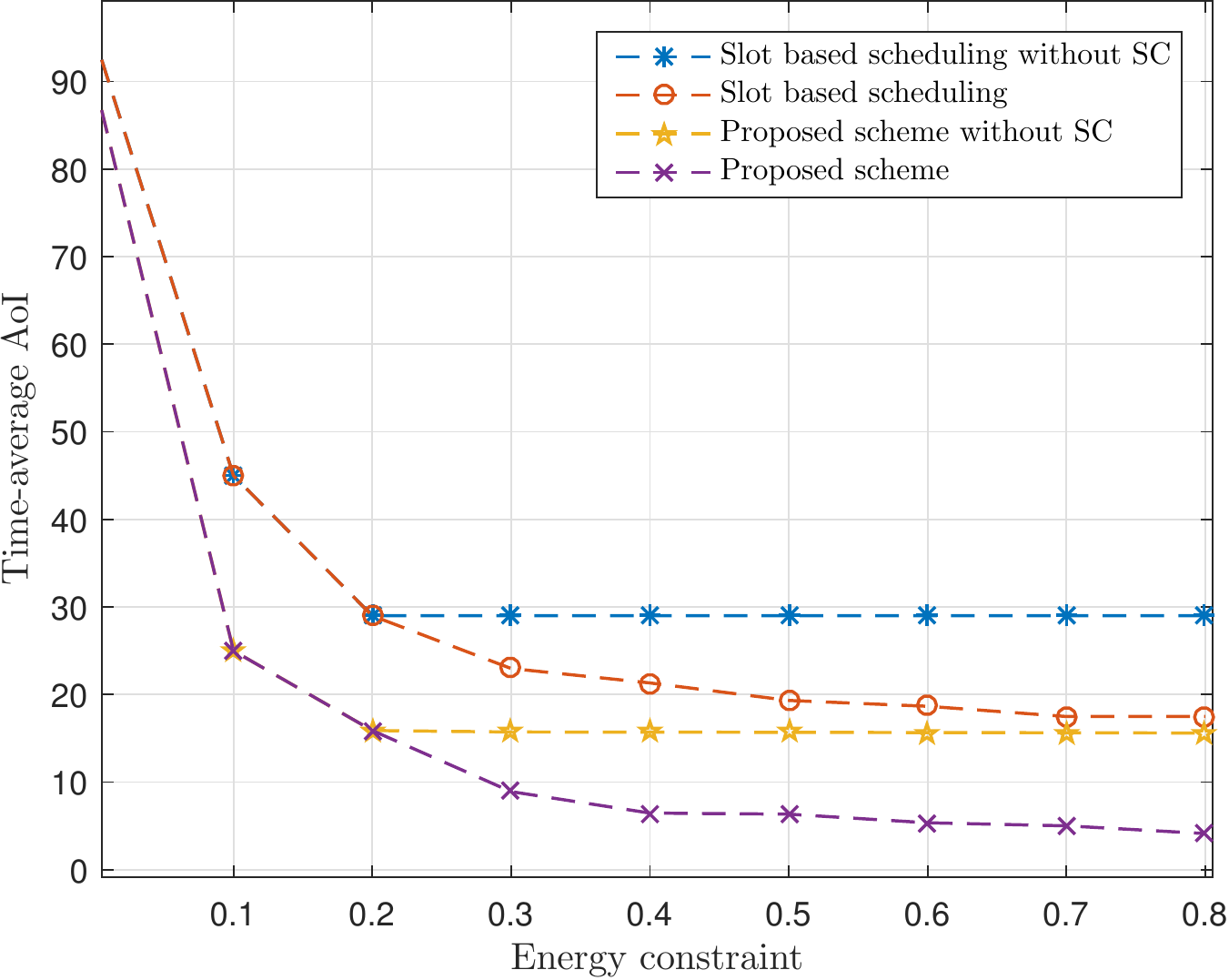}
	\caption{Time-average AoI performance as energy constraint $C^{max}$.}
	\label{fig6}
\end{figure}
\begin{figure}
	\centering
	\includegraphics[width=0.45\textwidth]{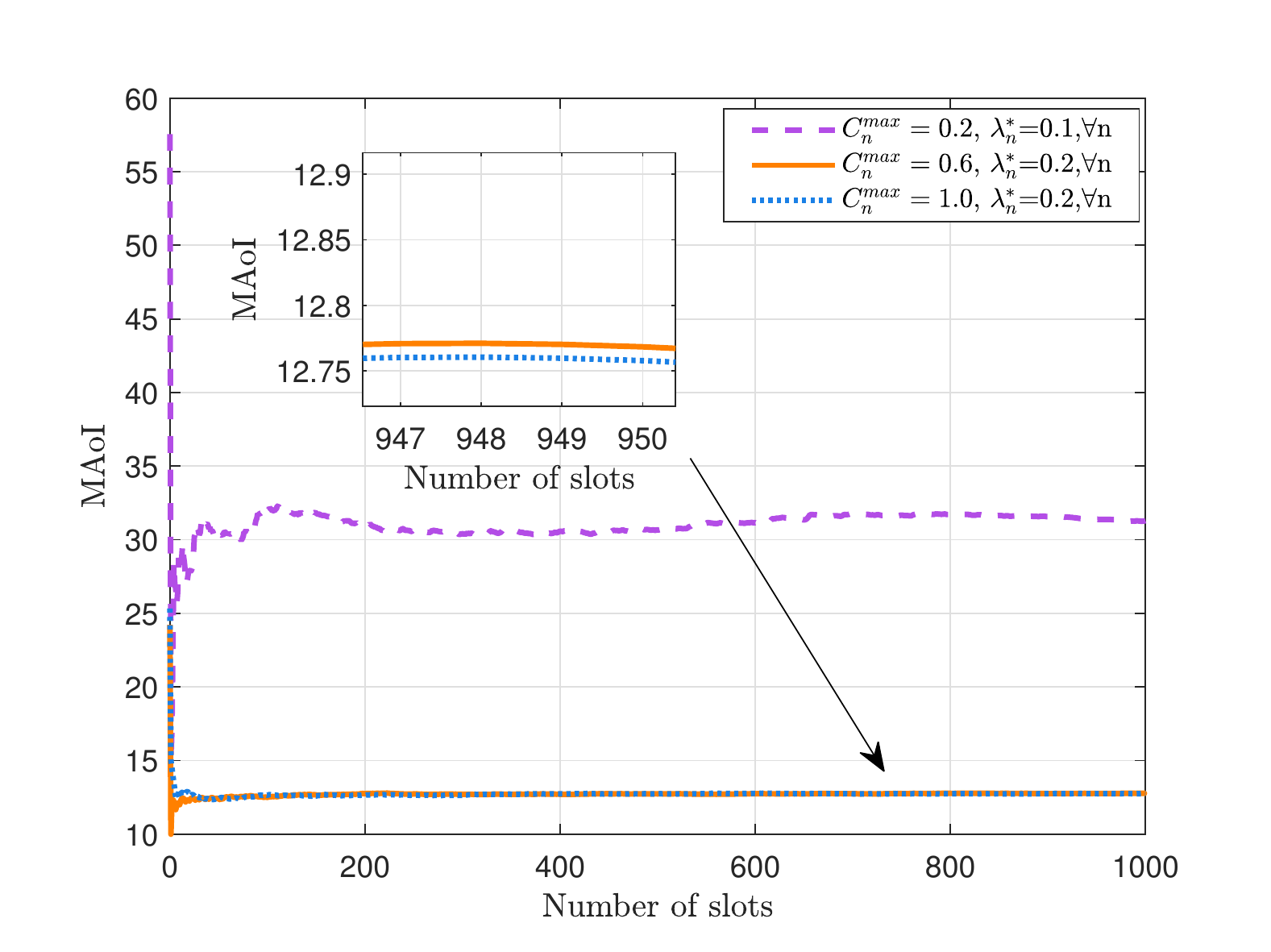}
	\caption{MAoI-energy trade-off and MAoI-queue stability trade-off.}
	\label{fig7}
\end{figure}

 Fig.~\ref{fig6} plots the AoI-energy tradeoff curves. The time-average AoI decreases monotonically with the increase of energy constraint, indicating the sensor has enough energy to transmit packets. When $C^{max}$ approaches zero, the time-average AoI approaches infinity. In addition, without control over the sampling rate, the time-average AoI changes little under the larger energy constraint. This is because the sampling rate of packets is fixed and the sufficient transmission capacity of the sensor is not fully utilized. Our proposed scheme consistently outperforms all the three baseline solutions, which adaptively make good use of energy to achieve a sustainable small reduction in the AoI.
 
\subsection{The Multi-sensor case}
This part evaluates the MAoI performance of a system with four sensors and each sensor has a queue capacity of 10 packets. Simulation results are obtained over continuous $10^3$ slots.
\begin{figure}
\centering
\includegraphics[width=0.45\textwidth]{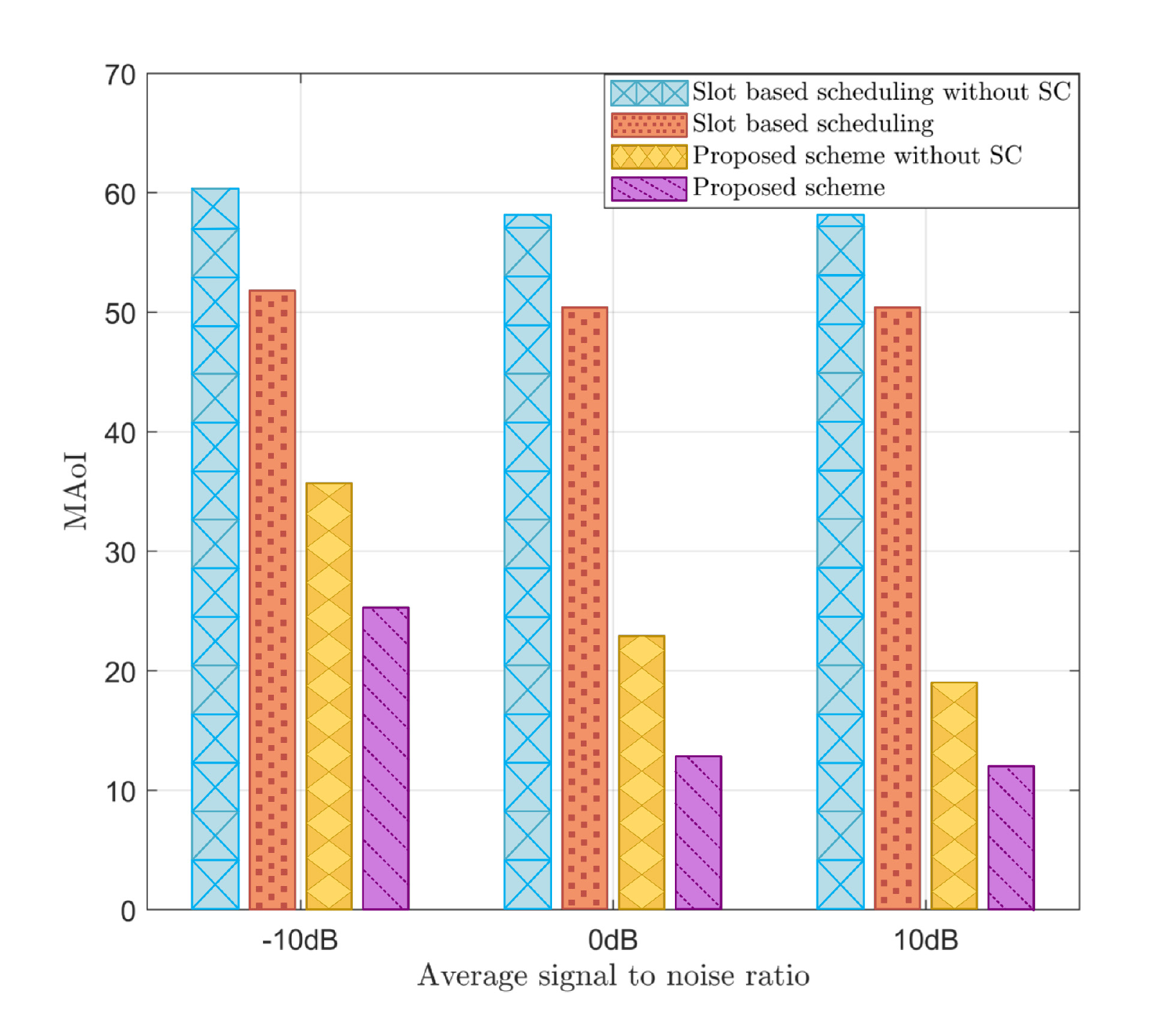}
\captionsetup{justification=centering} 
\caption{AoI performance as average \text{SNR}. ($C_n^{max}=0.6J$)}
\label{fig8}
\end{figure}
The effect of energy and queue stability constraint is seen directly from the optimal sampling rate and the performance of MAoI in Fig.~\ref{fig7}. When the primary reason for limiting sampling and transmission of packets is energy, abundant energy can realize a higher sampling rate and leads to a significant reduction in MAoI, e.g., blue and purple lines. We can attribute it to the reasons that there are more transmission opportunities under poor channel states. When the main reason is queue stability, more relaxed energy constraint no longer results in a remarkable decline in MAoI, e.g., blue and yellow lines, and the more powerful transmission capacity is not being fully utilized.
\begin{figure}
	\centering
	\includegraphics[width=0.45\textwidth]{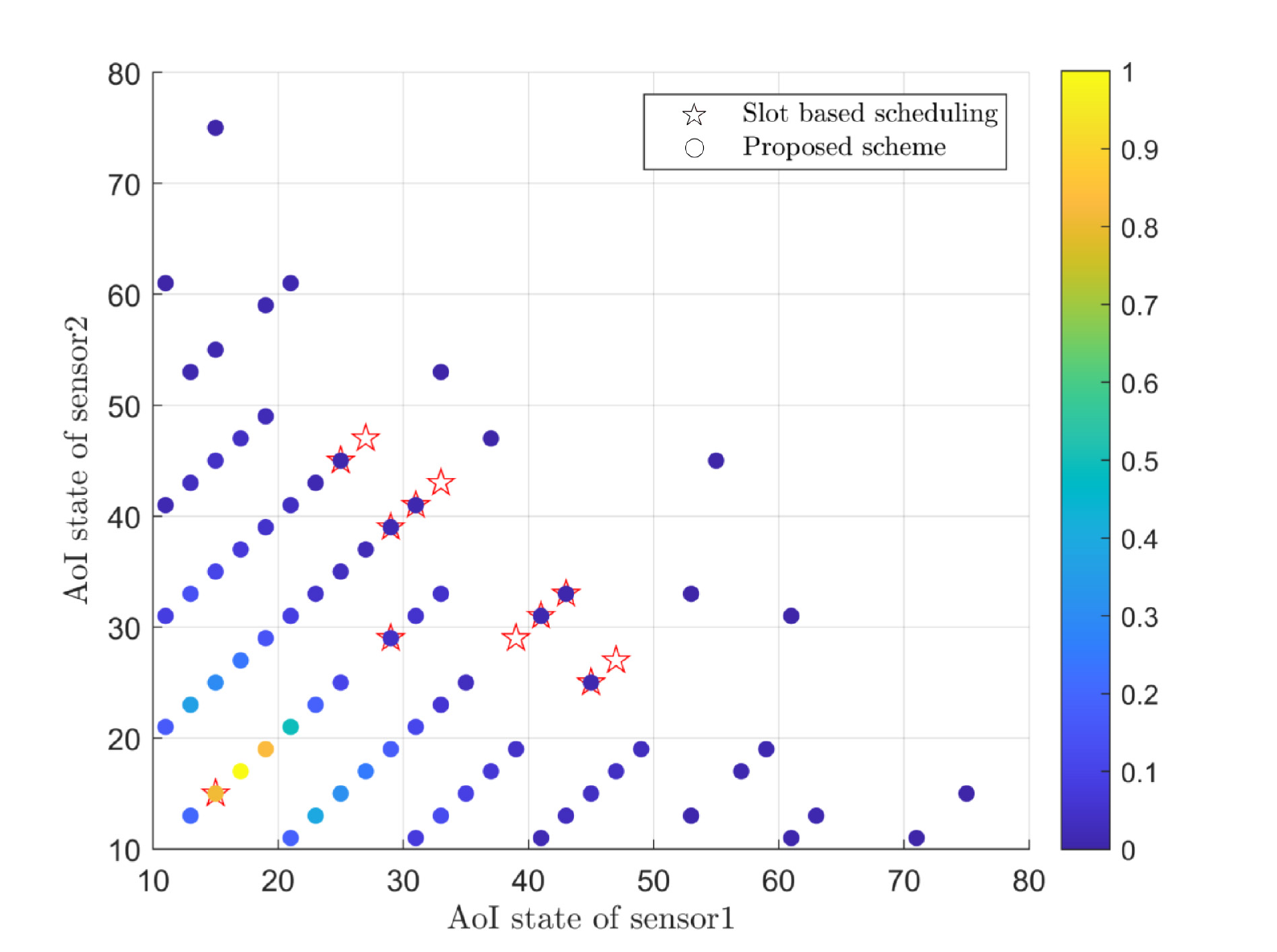}
	\caption{AoI performance of two sensors. ($C_n^{max}=0.2J$)}
	\label{fig9}
\end{figure}

Fig.~\ref{fig8} demonstrates the MAoI performance under different average $\text{SNR}$. As can be observed from the figure, for the better channel quality case, the energy consumption of transmitting packets becomes less. Sensors have more opportunities to be scheduled, resulting in a significant reduction in MAoI. However, this reduction is not obvious under the slot based scheduling without SC and the proposed scheme without SC. This is due to the fact that the packets of sensors are always generated at a low and fixed rate, which can not provide more data for transmission under the better channel state. Our proposed scheme can well adapt to the dynamic channel environment in time to increase the transmission probability and improve data freshness.

To show the performance of the proposed scheme in balancing and improving the data freshness, we depict the possible AoI state in a system of two sensors at the end of each slot in Fig.~\ref{fig9}. The lighter the color of the dot, the more frequently the state value appears. Intuitively, with our proposed scheme, the possible state values are more than the baseline, which happens when the channel state is bad. However, this unfair state does not occur many times, the frequency of the more uniform and smaller state value of the two sensors is higher. This is because once the channel state turns better, an adaptive amount of mini-slot resources can be allocated to the sensor with larger AoI in time to alleviate varying degrees of the imbalance between the two sensors in data freshness.

\begin{figure}
	\centering
	\includegraphics[width=0.45\textwidth]{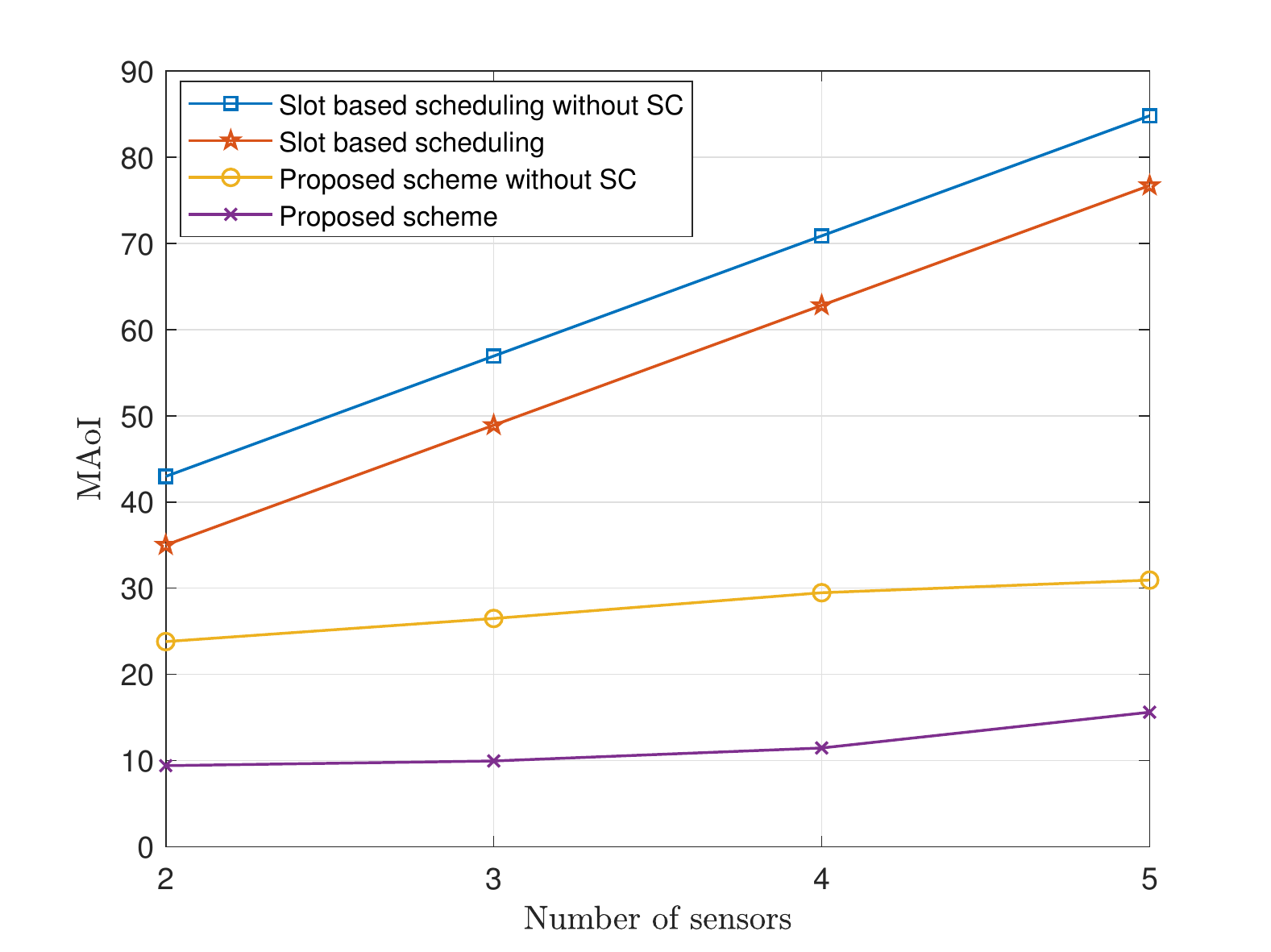}
	\caption{MAoI performance as number of sensors $N$. ($C_n^{max}=0.2J$)}\label{fig10}
\end{figure}

We evaluate the performance of our proposed scheme by comparing it with that of other baselines in Fig.~\ref{fig10}. It is seen that the data freshness becomes worse with the increase of the number of sensors. The reason behind is that the channel competition is more intense and the transmission opportunity of each sensor is reduced. What's more, the MAoI increases linearly with the number of sensors under the slot based scheduling, while the proposed scheme can maintain high data freshness and a relatively slow growth rate through fine-grained mini-slot resource allocation.
\section{Conclusion}
This paper studied the minimization of MAoI in the URLLC monitoring system. Each sensor sampled packets periodically and  sent them to the same destination through a time-varying channel with the constraints of energy and queue stability. An optimization framework of joint sampling and non-slot based scheduling policy was proposed. For single-sensor case, we first modelled the scheduling problem as a CMDP solved by RVI and sub-gradient descent method. Then we revealed the relationship between sampling rate and time-average AoI in terms of the steady-state distribution of the Markov chain. For multi-sensor case, a sub-optimal sampling and a semi-distributed scheduling scheme were proposed. The experimental results showed the effectiveness of the proposed scheme on reducing the MAoI. For future research, we will focus on more complex network models, such as random packets arrival and queue management.

\section{References Section}
\bibliographystyle{IEEEtran}
\bibliography{IEEEabrv,Ref}

\newpage

\vfill

\end{document}